%% file: MaxMinStochasticComputing.tex
\documentclass[12pt,journal,onecolumn,draftclsnofoot]{IEEEtran}

\usepackage[dvipsnames]{xcolor}	
\usepackage[T1]{fontenc}
\usepackage[latin9]{inputenc}
\usepackage{epsfig,psfrag}
\usepackage{graphicx}
\usepackage{amsmath,amsfonts,amssymb,amsxtra,bm,amsthm}
\usepackage{color}
\usepackage{paralist, tabularx}
\usepackage{subfigure}
\usepackage{balance}
\usepackage[adjust]{cite}
\usepackage{enumitem}
\usepackage{xfrac}
\usepackage[disable]{todonotes}  %
\usepackage{multirow}
\usepackage[final]{changes}
\usepackage{pifont}  
\usepackage{subfigure}
\usepackage{siunitx}
\usepackage[hidelinks]{hyperref}
\usepackage{multicol}

\usepackage{pgfplots, filecontents}
\usetikzlibrary{shapes}
\usepgfplotslibrary{patchplots}

\newcommand{\rmv}{\hspace*{-.3mm}}

\newtheorem{Theorem}{Proposition}
\newcommand{\fig}[1]{\hyperref[#1]{Fig.~\ref*{#1}}}
\newcommand{\tab}[1]{\hyperref[#1]{Tab.~\ref*{#1}}}
\newcommand{\equ}[1]{\hyperref[#1]{(\ref*{#1})}}
\newcommand{\alg}[1]{\hyperref[#1]{Alg.~\ref*{#1}}}
\newcommand{\app}[1]{\hyperref[#1]{App.~\ref*{#1}}}
\newcommand{\chap}[1]{\hyperref[#1]{Chapter~\ref*{#1}}}
\newcommand{\sect}[1]{\hyperref[#1]{Section~\ref*{#1}}}

\DeclareGraphicsExtensions{.pdf}
\definecolor{WernerRed}{RGB}{233,71,72}
\definecolor{WernerBlue}{RGB}{87,136,175}
\definecolor{WernerGreen}{RGB}{112,190,109}

\begin{document}

\title{\LARGE{Design and Analysis of Efficient Maximum/Minimum Circuits for Stochastic Computing}}

\author{Michael~Lunglmayr,~\IEEEmembership{Member,~IEEE,}
        Daniel~Wiesinger,
  Werner~Haselmayr,~\IEEEmembership{Member,~IEEE}
\IEEEcompsocitemizethanks{\IEEEcompsocthanksitem M. Lunglmayr and D. Wiesinger are with the Institute
of Signal Processing, Johannes Kepler University Linz, Austria.\protect\\
E-mail: michael.lunglmayr@jku.at
\IEEEcompsocthanksitem W. Haselmayr is with the Institute for Communications Engineering and RF-Systems, Johannes Kepler University Linz, Austria.}
\thanks{This work has been submitted to the IEEE for possible publication. Copyright may be transferred without notice, after which this version may no longer be accessible.}}

\maketitle

\input{Abstract}

\section{Introduction}
\label{sec:intro}
\input{Introduction}

\section{Stochastic Computing Basics}
\label{sec:basics}
\input{BasicsSC}

\section{State-of-the-Art Stochastic Max/Min Functions}
\label{sec:soa_max_min_fct}
\input{SoAMaxMinFct}

\section{Novel Stochastic Max/Min Function: Architecture}
\label{sec:new_max_min_fct_architec}
\input{NewMaxMinFctArchitecture}

\section{Novel Stochastic Max/Min Function: Error Analysis}
\label{sec:new_max_min_fct_analysis}
\input{NewMaxMinFctAnalysis}

\section{Conclusions}
\label{sec:concl}
\input{Conclusions}

\ifdefined\ACK
  \section{Acknowledgment}
  \input{Acknowledgment}
\fi


\bibliographystyle{IEEEtran}
\bibliography{IEEEabrv,References}

\end{document}

%% file: Abstract.tex

\begin{abstract}
  In stochastic computing (SC), a real-valued number is represented by a stochastic bit stream, encoding its value in the probability of obtaining a one. This leads to a significantly lower hardware effort for various functions and provides a higher tolerance to errors (e.g.,~bit flips) compared to binary radix representation. The implementation of a stochastic max/min function is important for many areas where SC has been successfully applied, such as image processing or machine learning (e.g.,~max pooling in neural networks). In this work, we propose a novel shift-register-based architecture for a stochastic max/min function. We show that the proposed circuit has a significantly higher accuracy than state-of-the-art architectures at comparable hardware cost. Moreover, we analytically proof the correctness of the proposed circuit and provide a new error analysis, based on the individual bits of the stochastic streams. Interestingly, the analysis reveals that for a certain practical bit stream length a finite optimal shift register length exists and it allows to determine the optimal length.
\end{abstract}

\begin{IEEEkeywords}
  Stochastic computing, sequential logic, finite state machine, stochastic max/min function
\end{IEEEkeywords}

%% file: Introduction.tex


%


\IEEEPARstart{S}{\lowercase{ochastic}} computing (SC) is a promising computing paradigm, 
which represents a real-valued number by a stochastic stream \cite{Gaines_69,Alaghi_13,Alaghi_17}. The value is encoded by the probability of obtaining a one in the stream. Compared to binary radix representation, the stochastic representation leads to low hardware cost and high fault tolerance to circuit noise and bit flips \cite{Li_14,Najafi_16}. 

In SC, basic arithmetic operations can be realized with simple combinational logic \cite{Gaines_69}. 
Moreover, combinational logic can be synthesized to implement arbitrary polynomial functions, by manipulating them into a Bernstein polynomial with coefficients in the unit interval \cite{Qian_11, Qian_11_1, Qian_08}. 
However, in order to implement more complex (non-linear) functions, sequential logic is required \cite{Brown_01, Li_14}. In particular, linear finite-state machines (FSM) have been proposed to implement complex functions, which can be realized by either 
employing saturating up/down counters or shift registers~\cite{Ting_17}. The stochastic exponentiation and the tanh function were presented in \cite{Brown_01} and the absolute value as well as exponentiation based on an absolute value were proposed in \cite{Li_14}. Recently, a new synthesis method which allows to implement arbitrary functions using FSM-based elements was introduced in \cite{Najafi_17}.

In recent years, SC has been successfully applied to a variety of applications such as decoding of modern error correcting codes \cite{LDPC_1,Dong_10,Lee_15}, control systems \cite{Marin_02,Zhang_08}, image processing~\mbox{\cite{Li_11, Alaghi_13_1,Li_14_1}}, filter design \cite{Ichihara_17,Chang_13}, and neural networks~\mbox{\cite{Ren_17,Yu_17,Brown_01,Brown_01_1}}. Most of these applications exploit the low complexity circuitry of SC in algorithms that do not require a high numerical precision of the final result.

In many of the aforementioned application domains, the efficient implementation of a stochastic max/min~(SMax/SMin) function is very important. Especially, for neural networks, where such functions are the key element in the max pooling layer \cite{Ren_17,Yu_17}. 
Two architectures for SMax functions\footnote{Both circuits can be easily converted to realize a SMin function.} have been proposed in literature~\mbox{\cite{Li_11, Yu_17}}. The implementation in \cite{Li_11} 
is based on a stochastic comparator, which requires a stochastic number generator (SNG) that is usually realized by a linear feedback shift register. In order to reduce the overhead of the SNG, an optimized SMax function was 
proposed in~\cite{Yu_17}. However, both approaches have only been validated empirically.

Thus, the first goal of this paper is to analytically prove the correctness of the SMax functions in \cite{Li_11, Yu_17}. Then, we propose a novel shift-register-based architecture for a stochastic SMax/SMin function and 
analytically prove its correctness. We show that the novel architecture provides a higher accuracy than \cite{Li_11, Yu_17} at comparable hardware cost. We provide a new error analysis of the proposed circuit, considering the individual bits of the stochastic streams. To the best of our knowledge, no such analysis has been done before for an FSM-based stochastic computing element. Based on the error analysis we show that 
for practical bit stream lengths a finite optimal shift register length exists. Moreover, we determine the optimal shift register size for certain bit stream lengths.

%% file: BasicsSC.tex

In this section, we briefly review the main principles of SC and introduce 
the basic computing elements used in this work. A comprehensive overview on SC can be found in \cite{Alaghi_13} and recent challenges 
and potential solutions are discussed in \cite{Alaghi_17}.

\subsection{Unipolar Coding Format} 
\label{subsec:unipolar_coding}

In the unipolar coding format\footnote{It is important to note that the circuits proposed in this work are also valid for the bipolar format, which enables the 
representation of negative values \cite{Gaines_69}.}, the value of a deterministic number $x \in [0,1]$ is encoded in a stochastic bit stream $X$ of length $N$. The individual bits in the stochastic 
stream are indicated by $X[i] \in \{0,1\}$. The probability for each bit in the stream to be one is given by $x = P_X = P(X[i] = 1)$. In practical realizations, the rate of ones in the stochastic bit stream is used to represent 
the number $x$
\begin{align}
  r(X) = \frac{\sum_{i=1}^N {X[i]}}{N} = \frac{o(X)}{N},
\end{align}

where $o(X) = \sum_{i=1}^N {X[i]}$ denotes the number of ones in the stream. The precision (representation resolution) of the unipolar format is given by $1/N$. 
Thus, $r(X) = x$ only if $N \rightarrow \infty$, otherwise $r(X)$ is only an approximation of~$x$. 


\begin{figure}[t!]
\centering
  \includegraphics[scale=0.7]{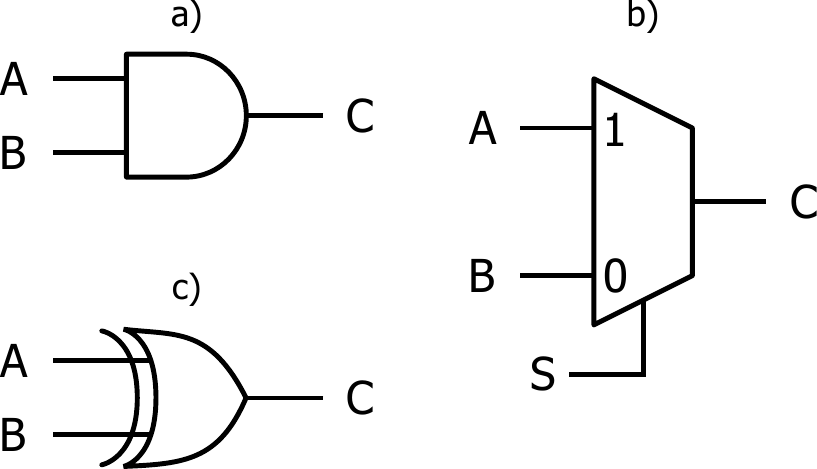}
  \caption{Stochastic circuits for (a) multiplication, (b) scaled-addition and (c) non-scaled addition and subtraction.}
  \label{fig:basic_sc_circuits}
\end{figure}
\subsection{Combinational Logic-based SC Elements} 
\label{subsec:comb_sce}

Certain arithmetic operations in SC can be implemented by single combinational elements, for example scaled addition and multiplication 
can be realized using a multiplexer and an AND gate, respectively. Moreover, an XOR gate with the input streams  $A$ and $B$
implements the function $A + B - 2AB$, which involves addition and subtraction. In the following, we briefly explain the principles of the aforementioned operations, 
as they are the main building blocks of the SMax/SMin functions presented in Secs.~\ref{sec:soa_max_min_fct}~ and~\ref{sec:new_max_min_fct_architec}.

\subsubsection{Multiplication}
\label{subsubsec:stoch_mult}
The stochastic multiplication can be implemented using a simple AND gate as shown in Fig. \ref{fig:basic_sc_circuits}a. If we assume that the 
input stochastic streams $A$ and $B$ are uncorrelated and have the probabilities $P_A$ and $P_B$, then we have at the output
\begin{align}
  P_C = P_A P_B.
  \label{eq:and_out_prob}
\end{align}

For the unipolar format the values encoded by the stochastic streams $A$, $B$ and $C$ are $a=P_A$, $b=P_B$ and $c=P_C$, and, thus we obtain
\begin{align}
   c =  ab.
  \label{eq:and_out}
\end{align}

\subsubsection{Scaled Addition}
\label{subsubsec:stoch_add}
The stochastic circuit for scaled addition, a multiplexer, is shown in Fig.~\ref{fig:basic_sc_circuits}b. If we assume that the stochastic streams $A$ and $B$ are uncorrelated with the stochastic stream $S$, and 
$P_A$, $P_B$~and~ $P_S$ are their corresponding probabilities, the output can be expressed as
\begin{align}
  P_C = P_S P_A + (1-P_S) P_B.
  \label{eq:mux_out_prob}
\end{align}

According to the unipolar format, we substitute $P_A$, $P_B$, $P_C$ and $P_S$ by $a$, $b$, $c$ and $s$, and obtain the following output
\begin{align}
  c = s a + (1-s) b.
  \label{eq:mux_out}
\end{align}

In order to perform unbiased addition, $s$ is set to~$1/2$.

\subsubsection{Non-Scaled Addition and Subtraction}
\label{subsubsec:stoch_add_sub}
If we assume uncorrelated input stochastic streams $A$ and $B$, with the corresponding probabilities $P_A$, $P_B$, then according to 
the Boolean function of the XOR gate (cf. Fig. \ref{fig:basic_sc_circuits}c) we have at the output 
\begin{align}
  P_C = P_A (1 - P_B) + P_B (1 - P_A) = P_A + P_B - 2P_AP_B.
  \label{eq:xor_out_prob}
\end{align}

For the unipolar coding format (i.e. $a = P_A$, $b=P_B$ and $c = P_C$)  we can rewrite \eqref{eq:xor_out_prob} as follows
\begin{align}
  c = a + b - 2ab.
  \label{eq:xor_out}
\end{align}

\begin{figure}[t!]
\centering
  \includegraphics[scale=0.7]{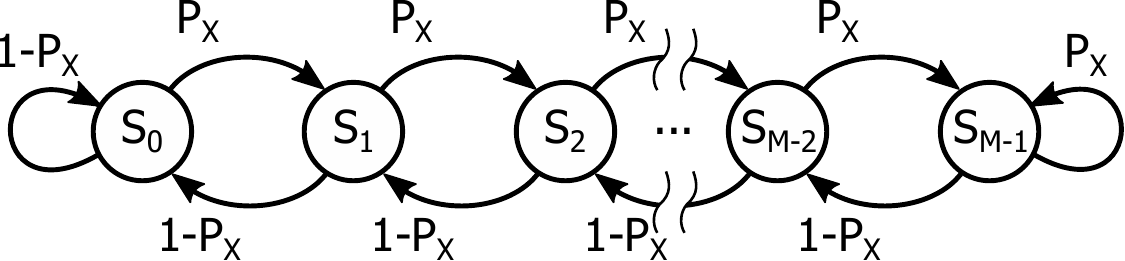}
  \caption{Linear finite state machine.}
  \label{fig:basic_lin_fsm}
\end{figure}
\subsection{FSM-based SC Elements} 
\label{subsec:fsm_sce}
Combinational logic can be used to realize polynomial functions of a specific form \cite{Qian_11} and to approximate non-polynomial functions, for example using the MacLaurin expansion \cite{Qian_08}. However, 
highly non-linear functions such as the exponential or the tanh function cannot be realized. Hence, FSM-based SC elements have been introduced~\cite{Brown_01,Li_14}. Here, we briefly review the 
stochastic tanh function\footnote{In \cite{Brown_01,Li_14} various other FSM-based SC elements are presented as well.}~(STanh) which builds the basis for the state-of-the art SMax/SMin functions  presented in~Sec. \ref{sec:soa_max_min_fct}.

\subsubsection{Stochastic Tanh Function}
\label{subsubsec:stanh}
Fig. \ref{fig:basic_lin_fsm} shows a linear FSM, with~$M$ states $S_0,\dots, S_{M-1}$ arranged in 
a linear form. The state transition process of the FSM can be modeled as a time-homogeneous irreducible and aperiodic Markov chain, which 
has a single steady state. The steady state probability is given by~\cite{Li_14}


\begin{align}
  P_i = \frac{\left(\frac{P_X}{1-P_X}\right)^i}{\sum\limits_{j=0}^{M-1} \left(\frac{P_X}{1-P_X}\right)^j},
  \label{eq:steady_state_prob}
\end{align}

where $P_X$ denotes the transition probability from state $S_i$ to $S_{i+1}$ (state is incremented) and $(1-P_X)$ indicates the transition 
from state $S_i$ to $S_{i-1}$ (state is decremented).

%

In order to realize the STanh function, the FSM output can be expressed as~\cite{Brown_01}
\begin{align}
  P_Z = \sum\limits_{i=M/2}^{M-1} P_i.
  \label{eq:fsm_output_tanh}
\end{align}

Substituting $P_i$ given in \eqref{eq:steady_state_prob} into \eqref{eq:fsm_output_tanh} results in \cite{Li_14}
\begin{align}
  P_Z 
      = \frac{\left(\frac{P_X}{1-P_X}\right)^{M/2}}{1 + \left(\frac{P_X}{1-P_X}\right)^{M/2}}.
  \label{eq:fsm_output_tanh_1}
\end{align}


If we substitute $P_X$ and $P_Z$ by $x$ and $z$ (unipolar coding format) we obtain~\cite{Li_14}
\begin{align}
  z & = \frac{\left(\frac{x}{1-x}\right)^{M/2}}{1 + \left(\frac{x}{1-x}\right)^{M/2}}
      =  \frac{1}{2} + \frac{\text{tanh}\left(M/2\left(x - 1/2\right)\right)}{2},
  \label{eq:stanh_unipolar}
\end{align}

which corresponds to a scaled and shifted tanh function\footnote{Substituting $P_X$ and $P_Z$ in \eqref{eq:fsm_output_tanh_1} with their bipolar coding format results in 
$z=\text{tanh}(M/2 x)$~\cite{Li_14}, representing the signum function for $M\rightarrow \infty$.}. For a large number of states $M$, \eqref{eq:stanh_unipolar}
behaves like a step function
\begin{align}
  \mathop{\text{lim}}_{M\rightarrow \infty} z = \begin{cases}
                                           0, & 0 \leq x < 0.5 \\
                                           0.5, &  x = 0.5 \\
                                           1, &  0.5 < x \leq 1.
         \end{cases}
\end{align}


%% file: SoAMaxMinFct.tex

In this section, we discuss two recently proposed architectures of SMax/SMin functions \cite{Li_11,Yu_17}. Moreover, we provide 
analytical proofs of their correctness, since~\mbox{\cite{Li_11,Yu_17}} only provide empirical validations.  For the sake of clearness, 
we focus our analysis on the SMax function, since the presented propositions and proofs can be easily applied to the SMin function.

  \begin{figure}[t!]
  \centering
    \includegraphics[scale=0.7]{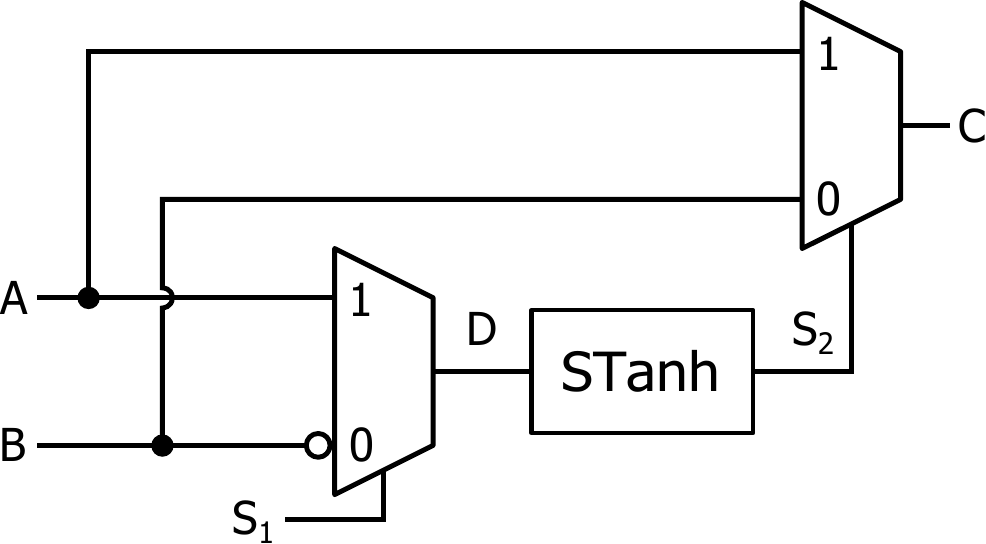}
    \caption{Implementation of the SMax function proposed in \cite{Li_11}}
    \label{fig:Max_1}
  \end{figure}

\subsection{Stochastic Max/Min Function in \texorpdfstring{\cite{Li_11}}{Lg}}
\label{subsec:max_min_fct_soa_1}
Fig. \ref{fig:Max_1} shows the architecture of the SMax function proposed in \cite{Li_11}, which is based on the stochastic comparator (input multiplexer and STanh function). 
The SMin function is obtained by swapping the input streams at the final multiplexer. The following proposition validates the correctness of the circuit shown in Fig. \ref{fig:Max_1}.
\begin{Theorem}
  For uncorrelated input bit streams $A$ and $B$, encoding the values $a = P_A$ and $b = P_B$ (unipolar coding format), the output of the circuit shown in Fig. \ref{fig:Max_1} can be expressed 
  as 
  \begin{align}
    c = a+ \frac{b - a}{1 + \left(\frac{1 + (a - b)}{1 + (b - a)}\right)^{M/2}},
  \label{eq:sc_max_unipolar_wang}
  \end{align}
  
  where $c= P_C$ denotes the value encoded in the output stream~$C$. For $M \rightarrow \infty$ the expression in \eqref{eq:sc_max_unipolar_wang} can be written as
  \begin{align}
    \max(a,b) = \mathop{\mathrm{lim}}_{M\rightarrow \infty} c =
                                             \begin{cases}
                                               a, & a > b \\
                                               a, & a=b\\
                                               b, & a < b,
                                             \end{cases}
    \label{eq:sc_max_limit_wang}
  \end{align}

  which validates the functionality of the SMax function.
\end{Theorem}

\begin{proof}
  The output of the first multiplexer is given by~(cf.~\eqref{eq:mux_out_prob})
  \begin{align}
    P_D = P_{S_1} P_A + (1-P_{S_1})(1-P_B) = 1/2(1 + P_A - P_B),
  \end{align}

  with $P_{S_1} = 1/2$. According to \eqref{eq:fsm_output_tanh_1}, the output of the STanh function can be expressed as
  \begin{align}
    P_{S_2} = \frac{\left(\frac{P_D}{1-P_D}\right)^{M/2}}{1 + \left(\frac{P_D}{1-P_D}\right)^{M/2}}.
  \end{align}
  

  Finally, the output of the second multiplexer is given by~(cf.~\eqref{eq:mux_out_prob}) 
  \begin{align}
    P_C & = P_{S_2}P_A + (1-P_{S_2})P_B \\
        & = P_A + \frac{P_B - P_A}{1 + \left(-1 + \frac{2}{1 - (P_A + P_B)}\right)^{M/2}}.
  \end{align}

  When substituting $P_A$, $P_B$ and $P_C$ with their corresponding unipolar coding format values $a$, $b$ and $c$ we obtain
  \begin{align}
    c = a+ \frac{b - a}{1 + \left(\frac{1 + (a - b)}{1 + (b - a)}\right)^{M/2}}.
  \end{align}

  If $M \rightarrow \infty$ the denominator in \eqref{eq:sc_max_unipolar_wang} becomes infinity or zero when $a>b$ or $a<b$, respectively. Thus, $c = a$ or $c = b$ if 
  $a>b$ or $a<b$, which proves the correctness of the SMax circuit shown in Fig. \ref{fig:Max_1}.
\end{proof}

Fig. \ref{fig:smax_wang} illustrates the analytical expression in \eqref{eq:sc_max_unipolar_wang}, the bit-wise simulation results of the circuit shown in Fig.~\ref{fig:Max_1} and the exact max function.
We observe a good match between the theoretical and simulation results. Moreover, we observe that already a moderate number of states $M$ provide a good approximation of the max function.
\todo[inline]{Fig.: Results for different states, Fig.: Verification of the analytical result, Fig./Table: Approx. Error}

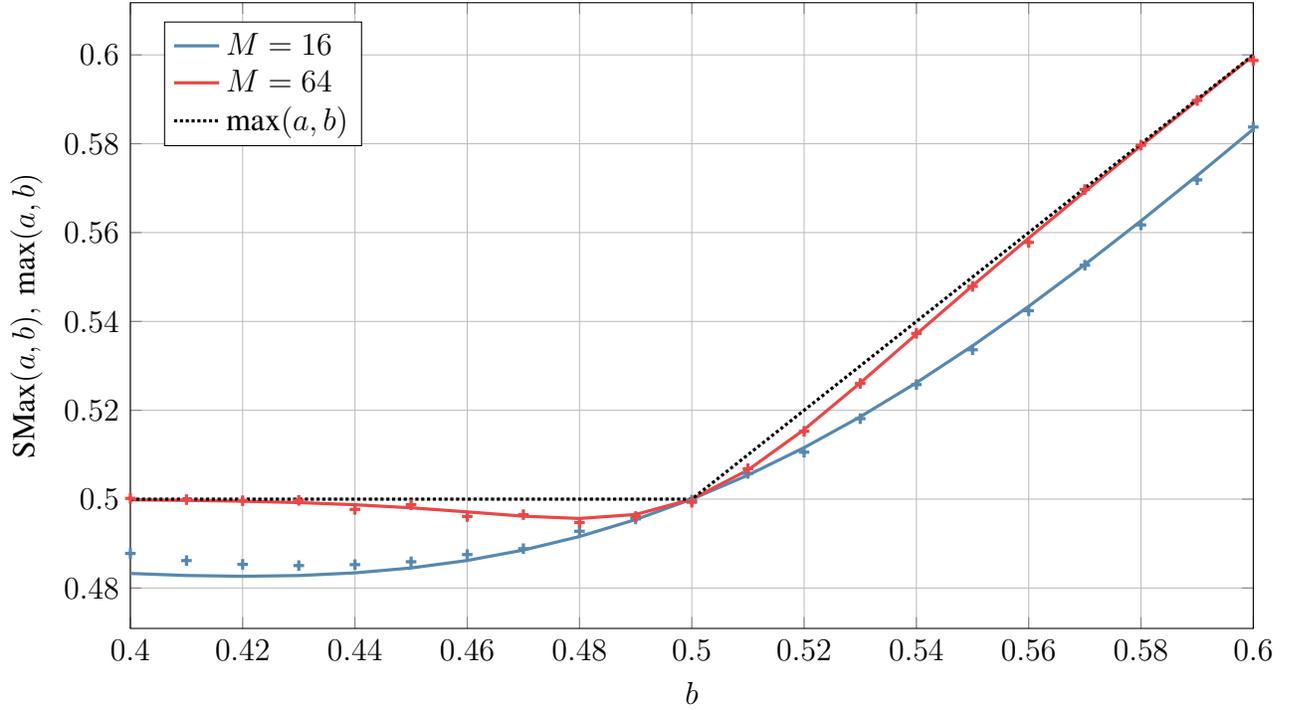
\begin{figure}[t]
\begin{center}
\vspace{-.03cm}
\begin{tikzpicture}
\begin{axis}[compat=newest, 
width=\columnwidth, height =.6\columnwidth,log basis y=10, grid,
ylabel={ SMax$(a,b),$ max$(a,b)$ }, 
xlabel={ $b$ }, 
xmin = 0.4,
xmax = 0.6,
legend pos=north west, 
legend columns = {1},
legend cell align=left,
]
\addplot[color=WernerBlue,very thick, only marks, mark=+, forget plot] table[x index =0, y index =2] {./StochasticMaxVarAsciiLi.dat};
\addplot[color=WernerBlue,very thick] table[x index =0, y index =3] {./StochasticMaxVarAsciiLi.dat};
\addlegendentry{ $M=16$ }
\addplot[color=WernerRed,very thick, only marks, mark=+, forget plot] table[x index =0, y index =4] {./StochasticMaxVarAsciiLi.dat};
\addplot[color=WernerRed,very thick] table[x index =0, y index =5] {./StochasticMaxVarAsciiLi.dat};
\addlegendentry{ $M=64$ }
\addplot[color=black, very thick, densely dotted] table[x index =0, y index =1] {./StochasticMaxVarAsciiLi.dat};
\addlegendentry{ max$(a,b)$ }
\end{axis}
\end{tikzpicture}
  \caption{SMax function \cite{Li_11}. Solid lines: theoretical results \eqref{eq:sc_max_unipolar_wang}; markers~(+): bit-wise simulation for $N = 10^{6}$; dotted line: exact max function.
  \label{fig:smax_wang}}
\end{center}
\end{figure}

\subsection{Stochastic Max/Min Function in \texorpdfstring{\cite{Yu_17}}{Lg}}
\label{subsec:max_min_fct_soa_2}
Fig. \ref{fig:Max_2} shows the architecture of the SMax function proposed in \cite{Yu_17}. Similar to Sec. \ref{subsec:max_min_fct_soa_1}, the SMin function is obtained by swapping the input streams at the final multiplexer. The following proposition validates the correctness of the circuit shown in Fig. \ref{fig:Max_2}.
\begin{Theorem}
  For uncorrelated input bit streams $A$ and $B$, encoding the values $a = P_A$ and $b = P_B$ (unipolar coding format), the output of the circuit shown in Fig. \ref{fig:Max_2} can be expressed 
  as
  \begin{align}
    c = a + \frac{b - a}{1 + \left(\frac{a(1-b)}{b(1-a)}\right)^{M/2}}.
    \label{eq:sc_max_unipolar_yu}
  \end{align}

  where $c= P_C$ denotes the value encoded in the output stream~$C$. For $M \rightarrow \infty$ the expression in \eqref{eq:sc_max_unipolar_yu} can be written as
  \begin{align}
    \max(a,b) = \mathop{\mathrm{lim}}_{M\rightarrow \infty} c = 
                                             \begin{cases}
                                               a, & a > b \\
                                               a, & a=b\\
                                               b, & a < b.
                                             \end{cases}
    \label{eq:sc_max_limit_yu}
\end{align}

  which validates the functionality of the SMax function.
\end{Theorem}

\begin{proof}
  The output of the XOR gate can be expressed as~(cf.~\eqref{eq:xor_out_prob})
  \begin{align}
    P_D = P_A + P_B - 2 P_A P_B.
    \label{eq:sc_max_xor_yu}
  \end{align}
  \begin{figure}[t!]
  \centering
    \includegraphics[scale=0.7]{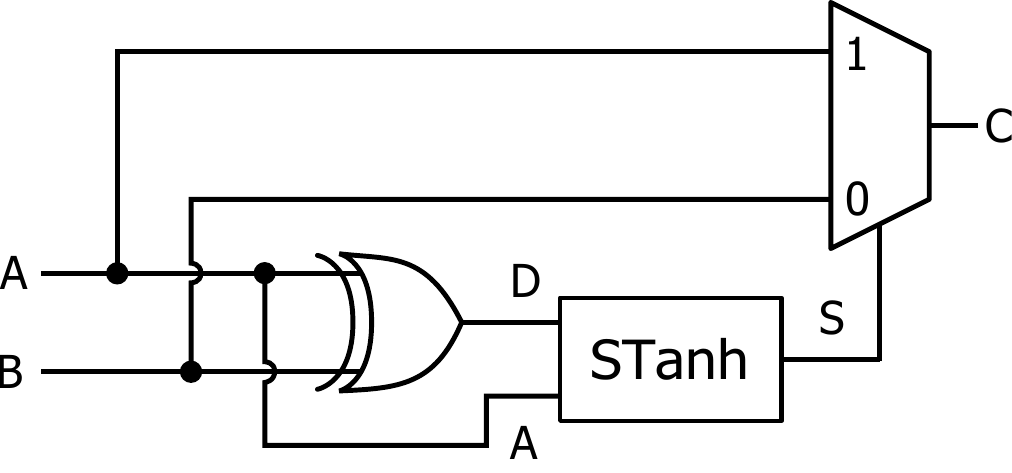}
    \caption{Implementation of the SMax function proposed in \cite{Yu_17}.}
    \label{fig:Max_2}
  \end{figure}

  \todo[inline]{Fig. of the new FSM}
  In contrast to the STanh function presented in Sec. \ref{subsubsec:stanh} the STanh function shown in Fig. \ref{fig:Max_2} has two inputs $A$ and $D$. The stochastic stream $D$ is used to enable the FSM state update and the $A$ updates the state according to its value. In particular, if $D[i]=1$ then the state increases if $A[i] = 1$ and decreases if $A[i] = 0$; if $D=0$ the state is not updated independently of $A[i]$. Thus, the probability that the state increases or decreases is given by $P_A P_D$ and $(1-P_A) P_D$, respectively. According to \eqref{eq:steady_state_prob}, the steady state probability is given by
  \begin{align}
    P_i & = \frac{\left(\frac{P_A P_D}{(1-P_A) P_D}\right)^i}{\sum\limits_{j=0}^{M-1} \left(\frac{P_A P_D}{(1-P_A) P_D}\right)^j} \nonumber \\
        & = \frac{\left(\frac{P_A(1-P_B)}{P_B(1-P_A)}\right)^i}{\sum\limits_{j=0}^{M-1} \left(\frac{P_A(1-P_B)}{P_B(1-P_A)}\right)^j},
    \label{eq:steady_state_prob_two_inp}
  \end{align}

  with $P_A P_D = P_A(1-P_B)$ and~\mbox{$P_D (1-P_A)  = P_B(1-P_A)$}. According to \eqref{eq:fsm_output_tanh_1}, the output of the STanh function can be expressed as
  \begin{align}
    P_S = \frac{\left(\frac{P_A(1-P_B)}{P_B(1-P_A)}\right)^{M/2}}{1 + \left(\frac{P_A(1-P_B)}{P_B(1-P_A)}\right)^{M/2}}.
  \end{align}

  Finally, the output at the multiplexer is given by~(cf.~\eqref{eq:mux_out_prob})
  \begin{align}
    P_C & = P_{S}P_A + (1-P_{S})P_B \\
        & = P_A + \frac{P_B - P_A}{1 + \left(\frac{P_A(1-P_B)}{P_B(1-P_A)}\right)^{M/2}}.
  \end{align}

  For the unipolar encoding format (i.e. $a = P_A$, $b = P_B$ and $c = P_C$) we have
  \begin{align}
    c = a + \frac{b - a}{1 + \left(\frac{a(1-b)}{b(1-a)}\right)^{M/2}}.
  \end{align}

  If $M \rightarrow \infty$ the denominator in \eqref{eq:sc_max_unipolar_yu} becomes infinity or zero when $a>b$ or $a<b$, respectively. Thus, $c = a$ or $c = b$ if 
  $a>b$ or $a<b$, which proves the correctness of the SMax circuit shown in Fig. \ref{fig:Max_2}.
  \end{proof}

Fig. \ref{fig:smax_yu} illustrates the analytical expression in \eqref{eq:sc_max_unipolar_yu}, the bit-wise simulation results of the circuit shown in Fig.~\ref{fig:Max_2}, and the exact max function.
We observe a good match between the theoretical and simulation results. Similar to Fig.~\ref{fig:smax_wang}, we observe that already a moderate number of states $M$ provide a good approximation of the max function. 
However, in Fig. \ref{fig:smax_yu} one can already see a closer match of this approach compared to the SMax function of \cite{Li_11}.

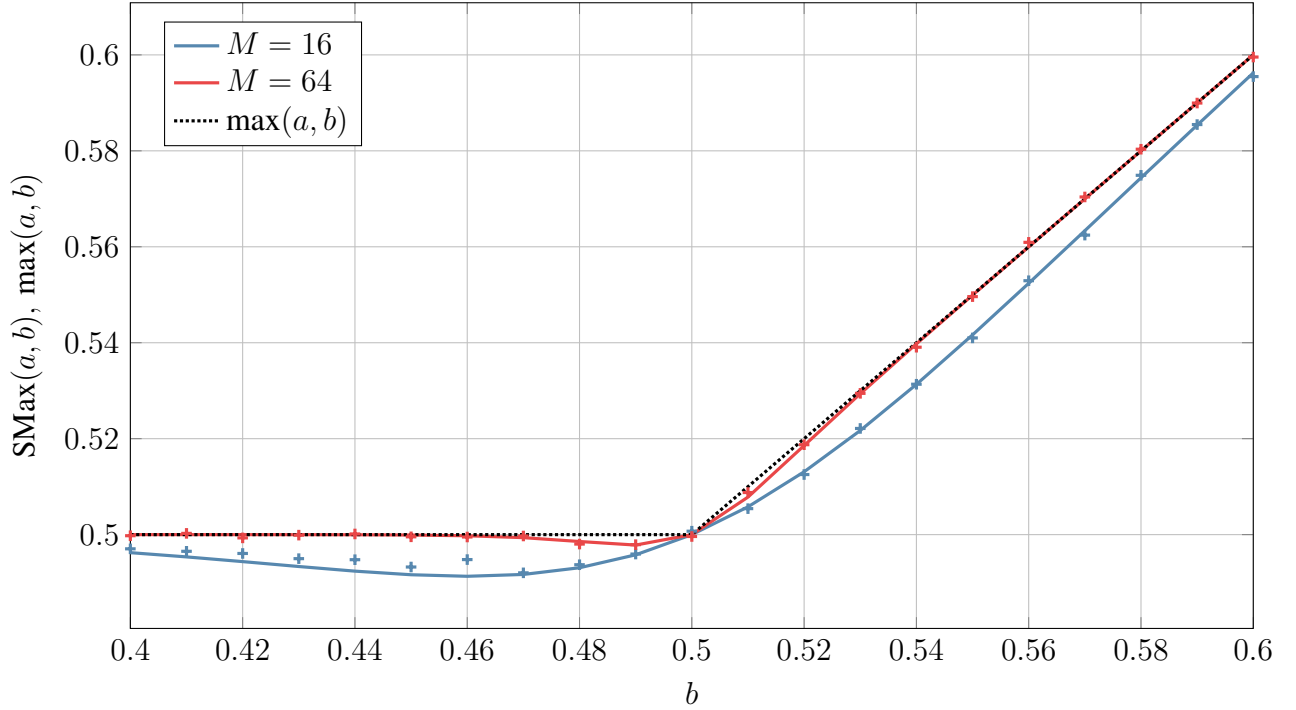
\begin{figure}[t]
\begin{center}
\vspace{-.03cm}
\begin{tikzpicture}
\begin{axis}[compat=newest, 
width=\columnwidth, height =.6\columnwidth,log basis y=10, grid,
ylabel={ SMax$(a,b),$ max$(a,b)$ }, 
xlabel={ $b$ }, 
xmin = 0.4,
xmax = 0.6,
legend pos=north west, 
legend columns = {1},
legend cell align=left,
]

\addplot[color=WernerBlue,very thick,  only marks, mark=+, forget plot] table[x index =0, y index =2] {./StochasticMaxVarAsciiYu.dat};
\addplot[color=WernerBlue,very thick] table[x index =0, y index =3] {./StochasticMaxVarAsciiYu.dat};
\addlegendentry{ $M=16$ }
\addplot[color=WernerRed,very thick,  only marks, mark=+, forget plot] table[x index =0, y index =4] {./StochasticMaxVarAsciiYu.dat};
\addplot[color=WernerRed,very thick] table[x index =0, y index =5] {./StochasticMaxVarAsciiYu.dat};
\addlegendentry{ $M=64$ }
\addplot[color=black,very thick, densely dotted] table[x index =0, y index =1] {./StochasticMaxVarAsciiYu.dat};
\addlegendentry{ max$(a,b)$ }
\end{axis}
\end{tikzpicture}
\caption{SMax function \cite{Yu_17}. Solid lines: Theoretical results \eqref{eq:sc_max_unipolar_yu}; markers~(+): bit-wise simulation for $N = 10^{6}$; dotted line: exact max function.  \label{fig:smax_yu}}
\end{center}
\end{figure}

%% file: NewMaxMinFctArchitecture.tex

In this section, we propose a novel architecture for the SMax function as shown in Fig.~\ref{fig:Max_new}. This circuit can be easily converted to realize the SMin function by inverting its inputs $A$ and $B$ as well as its output $C$. Hence, we only consider the SMax function in the following description.

In contrast to the state-of-the-art SMax functions \cite{Yu_17, Li_11}, the FSM-based SC element used in the proposed architecture does not implement the STanh function. However, similar to the architecture 
in \cite{Yu_17} it has two inputs $A$ and~$D$. Input~$D$ enables the FSM state update and $A$ updates the state according to its value (cf. Sec. \ref{subsec:max_min_fct_soa_2}). FSM-based elements 
can either be  implemented using up/down counters or shift registers. When using a shift register, its length $L$ is 
equal to the last state of the FSM, i.e. $L = M-1$.
For the novel SMax function we use a shift register, since it has some distinct advantages compared to a counter-based implementation~\cite{Ting_17}. One advantage is that the values in a shift register are 
of equal significance, in contrast to a binary counter, where the bits are weighted by different powers of two. This allows to design more fault-tolerant SC computing circuits when using shift registers. Furthermore, as the following description demonstrates, shift registers are naturally suited to implement the described functionality. 

Depending on the actual value of the input streams $A$~and~$B$ the functionality of the SMax function (cf.  Fig.~\ref{fig:Max_new}) can be described as follows:
\begin{itemize}
  \item $A[i] \rmv  = \rmv B[i]$: Since $D[i] = 0$ the content of the shift register remains unchanged (state is not updated) and the output of the circuit is given by $C[i] = B[i]$.
  \item $A[i] \rmv = \rmv 0, B[i] \rmv = \rmv 1$:  Since $D[i] = 1$ and $A[i]=0$ a zero is shifted from the right into the shift register (state is decremented) and the the output of the circuit is given by $C[i] = B[i]$.
  \item $A[i] = 1, B[i] = 0$: Since $D[i] = 1$ and $A[i]=1$ a one  is shifted from the left into the shift register (state is incremented). In this case the rightmost value of the shift register is output by the circuit, i.e. $C[i] = U[i]$.  
\end{itemize}
According to the description above it is important to note that the ones in the stream $B$ also appear in the output stream $C$.

In the following, we provide two approaches for analyzing the functionality of the proposed SMax function. First, we describe the functionality by considering the individual bits in the stochastic bit stream. 
Then, similar to Sec. \ref{sec:soa_max_min_fct} we proof the correctness of the circuit assuming very long stochastic bit streams. We denote these two methods as deterministic and probabilistic analysis, respectively.

\begin{figure}[t!]
\centering
  \includegraphics[scale=0.7]{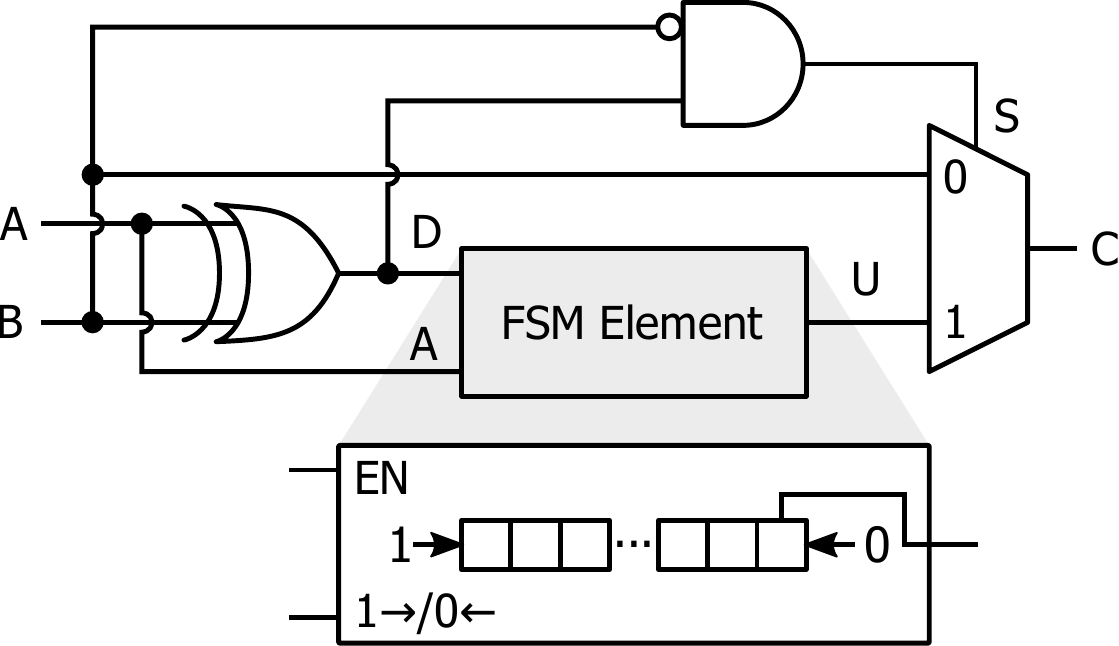}
  \caption{Implementation of the novel SMax function.}
  \label{fig:Max_new}
\end{figure}



\subsection{Deterministic Analysis}
\label{subsec:novel_smax_det}

For the deterministic analysis of the novel SMax function we distinguish the two cases: $r(A) \leq r(B)$ and $r(A)>r(B)$.

\subsubsection{SMax Circuit Behavior for \texorpdfstring{$r(A) \leq r(B)$}{Lg}}

If $r(A) \leq r(B)$, the stream $B$ has more (or equal) ones than bit stream $A$. For a correct functionality it is desired that the number of ones $o(B)$ in the input stream $B$,  and the number of ones $o(C)$ in the output stream $C$, are equal. As discussed above, all ones of stream $B$ are included in the output stream $C$. However, if a subsequence of $A$ has more ones than the corresponding subsequence of $B$, also ones of stream $A$ might be additionally injected into the output stream $C$. This occurs if the excess of ones in this subsequence is larger than the shift register length $L$. We refer to such an event as right overflow of the shift register. Thus, the number of ones in the output stream~$C$ can be expressed as
\begin{align}
  o(C) = o(B) + o_\text{R},
\end{align}

where $o_R$ denotes the additional number of ones due to the right overflows. If the shift register is sufficiently long, no right overflows occur, i.e. $o(C) = o(B)$.

\subsubsection{SMax Circuit Behavior for \texorpdfstring{$r(A)>r(B)$}{Lg}}
\label{}

If $r(A)>r(B)$ , the bit stream $A$ has more ones than bit stream $B$. For a correct functionality, it is desired that $o(A)$, the number of ones in the input stream $A$, and $o(C)$, the number ones in the output stream $C$, are identical. Similar as above, all ones of $B$ are included in the output stream $C$. In addition, 
the excess of ones in stream $A$ is shifted into the shift register and once the shift register is filled, the ones are injected into the output stream $C$ when $B[i] = 0$ and $A[i] =1$ occurs. 
However, at the end there might be ones left in the shift register, which are missing in the output stream\footnote{We assume the same length for all stochastic streams, which is a typically assumption in SC.} $C$. We denote the number of missing ones by $o_\text{S}$. 

Moreover, if a subsequence of $B$ has more ones than the corresponding  subsequence of $A$, also ones of stream $B$ might additionally be injected into the output stream $C$. This happens if the excess of ones leads to an empty (all-zero) shift register, and, thus, 
an input pattern $B[i]=1$ and $A[i]= 0$ (and assuming a later following $B[i]=0$ and $A[i]= 1$) injects an additional one in the output stream $C$. We refer to this effect as left overflow of the shift register and denote the number of additional ones due to left overflows by $o_\text{L}$. This allows expressing the number of ones in the output stream~$C$ as
\begin{align}
  o(C) = o(B) + (o(A) - o(B)) + o_\text{L} - o_\text{S},
  \label{eqn:ovandremain}
\end{align}
where $(o(A) - o(B))$ denotes the excess of ones in stream~$A$ compared to stream~$B$. Expression (\ref{eqn:ovandremain}) shows the two opposite 
error effects. One the one hand, the number of left overflows $o_\text{L}$ becomes smaller for long shift registers. On the other hand, the error due to the remaining ones in the shift register $o_\text{S}$ becomes smaller, for short shift registers. In Sec.~\ref{sec:new_max_min_fct_analysis} we determine the optimal shift register length for a given bit stream length.

\subsection{Probabilistic Analysis}
\label{subsec:novel_smax_prob}
The following proposition validates the correctness of the circuit shown in Fig. \ref{fig:Max_new}.
\begin{Theorem}
  For uncorrelated input bit streams $A$ and $B$, encoding the values $a = P_A$ and $b = P_B$ (unipolar coding format), the output of the circuit shown in Fig. \ref{fig:Max_new} can be expressed 
  as
  \begin{align}
    c = b + \frac{b - a}{\left(\frac{b\left(1-a\right)}{a\left(1-b\right)}\right)^M - 1}.
    \label{eq:sc_max_unipolar_new}
  \end{align}

  where $c= P_C$ denotes value encoded in the output stream~$C$. For $M \rightarrow \infty$ the expression in \eqref{eq:sc_max_unipolar_new} can be written as
  \begin{align}
    \max(a,b) = \mathop{\text{lim}}_{M\rightarrow \infty} c = 
                                             \begin{cases}
                                               c = a, & a > b \\
                                               c = b, & a=b\\
                                               c = b, & a < b.
                                             \end{cases}
    \label{eq:sc_max_limit_new}
  \end{align}

  which validates the functionality of the SMax function.
\end{Theorem}

\begin{proof}
  The output of the XOR gate can be expressed as~(cf.~\eqref{eq:xor_out_prob})
  \begin{align}
    P_D = P_A + P_B - 2 P_A P_B.
    \label{eq:sc_max_xor_new}
  \end{align}
  
  Similar to Sec. \ref{subsec:max_min_fct_soa_2}, the FSM has two inputs $A$ and $D$ and, thus, the steady state probability can be written as (cf. \eqref{eq:steady_state_prob_two_inp})
  \begin{align}
    P_i & = \frac{\left(\frac{P_A(1-P_B)}{P_B(1-P_A)}\right)^i}{\sum\limits_{j=0}^{M-1} \left(\frac{P_A(1-P_B)}{P_B(1-P_A)}\right)^j},
    \label{eq:sc_max_fsm_state_general_new}
  \end{align}

  According to Fig.~\ref{fig:Max_new}, the FSM outputs can be expressed as
  \begin{align}
    P_U & = P_{M-1} = \frac{\left(\frac{P_A(1-P_B)}{P_B(1-P_A)}\right)^{M-1}}{\sum\limits_{j=0}^{M-1} \left(\frac{P_A(1-P_B)}{P_B(1-P_A)}\right)^j} \nonumber \\
        & = \frac{P_B - P_A}{P_A(1-P_B)\left(-1+{\left(\frac{P_A\left(1-P_B\right)}{P_B\left(1-P_A\right)}\right)^{-M}}\right)}.
    \label{eq:sc_max_fsm_state_last_new}
  \end{align}
  

  The output of the AND gate can be calculated as~(cf.~\eqref{eq:mux_out_prob})
  \begin{align}
    P_S = P_D(1-P_B) = P_A(1-P_B).
  \end{align}

  Finally, the output of the multiplexer is given by
  \begin{align}
    P_C & = P_S P_U + (1-P_S)P_B \nonumber \\
        & =P_B + \frac{P_B - P_A}{\left(\frac{P_B\left(1-P_A\right)}{P_A\left(1-P_B\right)}\right)^M - 1}
  \end{align}

  For the unipolar encoding format (i.e. $a = P_A$, $b = P_B$ and $c = P_C$) we have
  \begin{align}
    c = b + \frac{b - a}{\left(\frac{b\left(1-a\right)}{a\left(1-b\right)}\right)^{M} - 1}.
  \end{align}
  
If $M \rightarrow \infty$ the denominator in \eqref{eq:sc_max_unipolar_yu} becomes zero or infinity depending on whether $a>b$ or $a \leq b$, respectively. Thus, $c = a$ or $c = b$ if 
$a>b$ or $a \leq b$, which proves the correctness of the SMax circuit shown in Fig. \ref{fig:Max_new}.
\end{proof}

Fig. \ref{fig:smax_new} illustrates the analytical expression in \eqref{eq:sc_max_unipolar_new}, the bit-wise simulation results of the circuit shown in Fig.~\ref{fig:Max_new} and the exact max function.
We observe a good match between the theoretical and simulation results. Moreover, we observe that already a low number of states $M$ provide a good approximation of the max function. 
In contrast to the state-of-the-art SMax functions \cite{Li_11,Yu_17} the proposed function does not approach the exact max function at $a=b$, but provides a better approximation for $a\neq b$.

Next, we compare the approximation error of the state-of-the-art SMax functions and the novel SMax function. 
For this we calculate the expected value of the absolute error, assuming a uniform distribution of $a$ and $b$ over the interval~$[0,1]$, respectively. 
We define the absolute error $e$ by $e=|c_\text{exact} - c|$, with the exact max function $c_\text{exact} = \max(a,b)$ and the FSM-based approximations $c$ given in \eqref{eq:sc_max_unipolar_wang}, \eqref{eq:sc_max_unipolar_yu} and~\eqref{eq:sc_max_unipolar_new}, respectively. Then, the expected absolute error can be calculated as 
\begin{align}
  E(e) & = \int\limits_0^1 \int\limits_0^1 |c_\text{exact} - c| \text{d}a \text{d}b.
  \label{eqn:ExpMInf}
\end{align}
The absolute value of the error in (\ref{eqn:ExpMInf}) allows to consider both,
erroneously added ones (i.e. $c > c_\text{exact}$) as well as erroneously removed ones (i.e. $c < c_\text{exact}$) in the bit stream representing $c$. When considering a stochastic bit stream of~$c_\text{exact}$, then the absolute difference $e$ can be interpreted as a bit error probability of $c$ compared to such a bit stream of $c_\text{exact}$. 
This is crucial in order to enable a comparison with the analysis results 
presented in Sec.~\ref{sec:new_max_min_fct_analysis}.
We observe 
from Fig. \ref{fig:appr_error_smax} that the novel SMax function has a significantly lower approximation error than the state-of-the-art SMax functions. This is because the denominator of \eqref{eq:sc_max_unipolar_new}~\mbox{(power of $M$)} converges faster to zero or infinity as the number of states~$M$ increases compared to the denominators in \eqref{eq:sc_max_unipolar_wang} or \eqref{eq:sc_max_unipolar_yu} (power of $M/2$).
Moreover, we observe that the SMax function proposed in \cite{Li_11} has the highest approximation error among the 
three approaches.


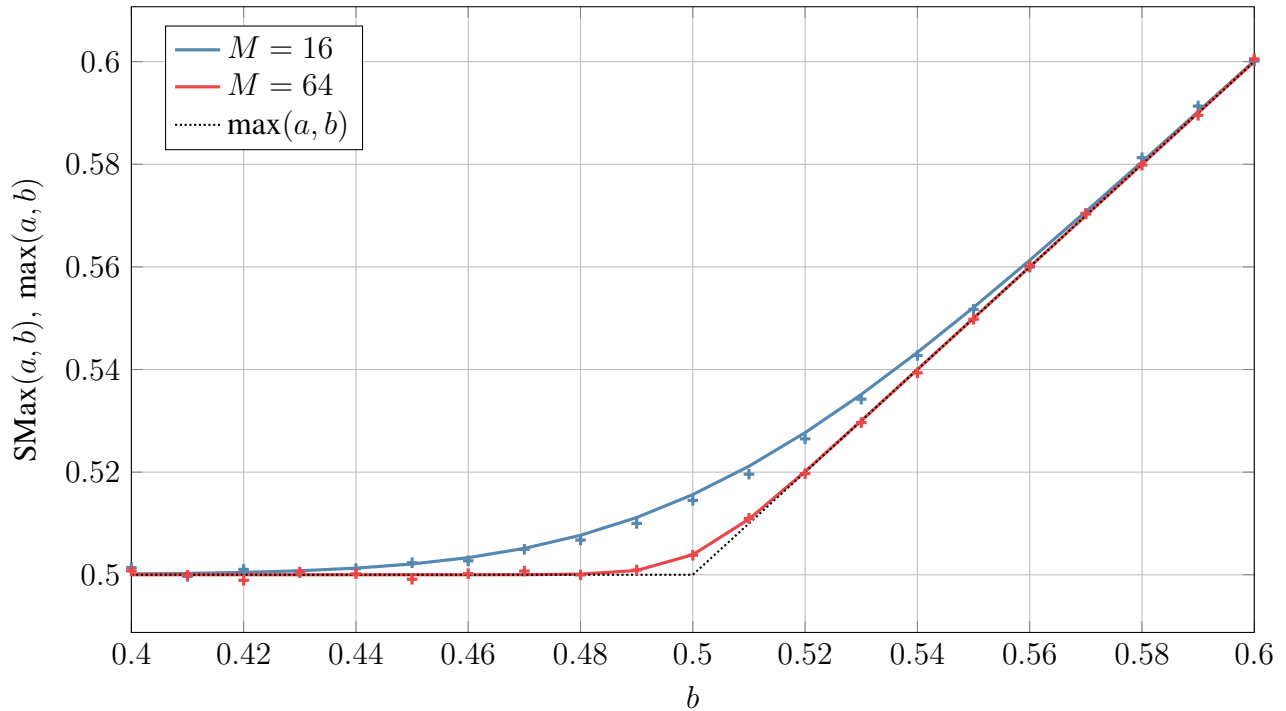
\begin{figure}[t]
\begin{center}
\vspace{-.03cm}
\begin{tikzpicture}
\begin{axis}[compat=newest, 
width=\columnwidth, height =.6\columnwidth,log basis y=10, grid,
ylabel={ SMax$(a,b),$ max$(a,b)$ }, 
xlabel={ $b$ }, 
xmin = 0.4,
xmax = 0.6,
legend pos=north west, 
legend columns = {1},
legend cell align=left,
]

\addplot[color=WernerBlue,very thick,  only marks, mark=+, forget plot]table[x index =0, y index =2] {./StochasticMaxVarAsciiNew.dat};
\addplot[color=WernerBlue,very thick]table[x index =0, y index =3] {./StochasticMaxVarAsciiNew.dat};
\addlegendentry{ $M=16$ }
\addplot[color=WernerRed,very thick,  only marks, mark=+, forget plot]table[x index =0, y index =4] {./StochasticMaxVarAsciiNew.dat};
\addplot[color=WernerRed,very thick]table[x index =0, y index =5] {./StochasticMaxVarAsciiNew.dat};
\addlegendentry{ $M=64$ }
\addplot[color=black,thick, densely dotted] table[x index =0, y index =1] {./StochasticMaxVarAsciiNew.dat};
\addlegendentry{ max$(a,b)$ }
\end{axis}
\end{tikzpicture}
  \caption{Novel SMax function. Solid lines: theoretical results \eqref{eq:sc_max_unipolar_new}; markers~(+): bit-wise simulation of Fig. \ref{fig:Max_2} for $N = 10^{6}$; dotted line: exact max function.
  \label{fig:smax_new}}
\end{center}
\end{figure}

\begin{figure}[t]
\begin{center}
\vspace{-.03cm}
\begin{tikzpicture}
\begin{semilogyaxis}[compat=newest, 
width=\columnwidth, height =.6\columnwidth,log basis y=10, grid,
ylabel={ $E(e)$ }, 
xlabel={ States $M$ }, 
xmin = 0,
xmax = 50,
legend pos=north east, 
legend columns = {1},
legend cell align=left,
]
\addplot[color=WernerBlue,very thick] table[x index =0, y index =1] {./ApproxErrorAnalysisVarAscii_51.dat};
\addlegendentry{ SMax of \cite{Li_11} }
\addplot[color=WernerGreen,very thick] table[x index =0, y index =2] {./ApproxErrorAnalysisVarAscii_51.dat};
\addlegendentry{ SMax of \cite{Yu_17} }
\addplot[color=WernerRed,very thick] table[x index =0, y index =3] {./ApproxErrorAnalysisVarAscii_51.dat};
\addlegendentry{ Novel SMax }
\end{semilogyaxis}
\end{tikzpicture}
\caption{Expected absolute error versus different number of states.
  \label{fig:appr_error_smax}}
\end{center}
\end{figure}
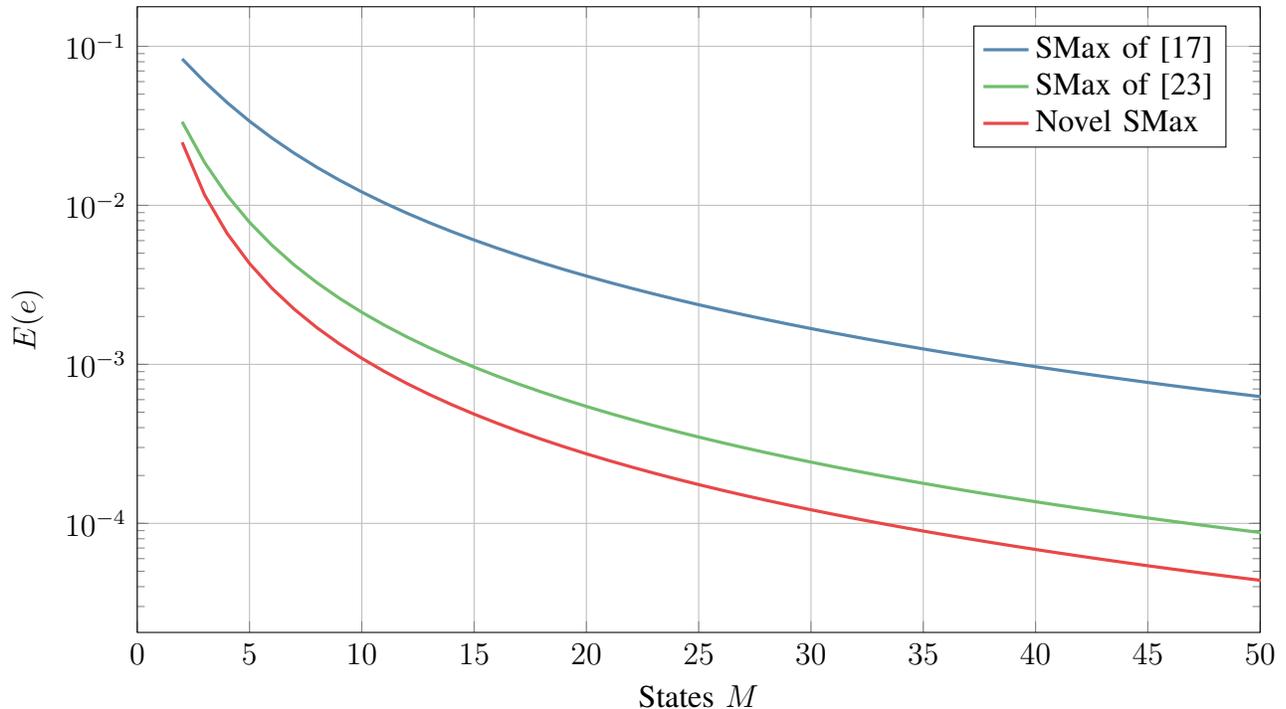

%% file: NewMaxMinFctAnalysis.tex

We observed from the deterministic analysis in Sec.~\ref{subsec:novel_smax_det} that  
long shift registers reduce the errors due to right and left overflows. However, a large shift register 
length increases the error caused by the remaining bits in the shift register. It is important to note that especially the 
last observation cannot be inferred from the probabilistic analysis presented in the previous chapter. The probabilistic analysis becomes exact if and only if the bit stream length goes to infinity. In such a case, a finite number of remaining bits in the shift register does not matter. However, when using a finite bit stream length, these remaining bits matter. In the following analysis, we assume a finite
bit stream length that is significantly larger than the shift register length (a typical scenario in SC). This allows using the probabilistic FSM description in Sec.~\ref{subsec:novel_smax_prob} for modelling the behavior of the shift register for finite bit stream lengths. 
Moreover, having sufficiently long bit streams justifies using the probabilities of ones instead of the rate of ones in the stream  (cf. Sec.~\ref{subsec:unipolar_coding}) for the following error analysis.

In the following, we derive the expected error probability, based on the deterministic analysis in Sec.~\ref{subsec:novel_smax_det}. 
With this expression we determine the optimal shift register length~$L_\text{opt}$, with respect to the bit stream length $N$.
Similar to Sec.~\ref{subsec:novel_smax_det} we distinguish two cases: $a \leq b$ and $a > b$. 

\subsection{Error Probability for  \texorpdfstring{$a \leq b$}{Lg}}
In this case, an error occurs due to right overflows. In particular, the shift register is filled with ones, i.e. the FSM is in the last state $M-1$, and the input 
$A[i] = 1$ and $B[i] = 0$ is applied. The probability for this error event can be described as
\begin{align}
  P_{e, a \leq b} = P_{M-1} P_A (1-P_B),
\label{eqn:Peab}
\end{align}

where $P_{M-1}$ describes the probability of the FSM to be in the last state (cf. \eqref{eq:sc_max_fsm_state_last_new}).

\subsection{Error Probability for  \texorpdfstring{$a > b$}{Lg}}
In this case, errors can originate from two sources: Left overflow and remaining ones in the shift register.
At the left overflow, the shift register is empty (all-zeros), i.e. the FSM is in state $S_0$, and the input pattern $A[i] = 0$ and $B[i] = 1$ occurs. 
The probability for this error event can be expressed as
\begin{align}
  P_{e,0} = P_0 P_B (1-P_A).
\end{align}

where $P_0$ denotes the probability of the zero state of the FSM, i.e. an all-zero shift register (cf. \eqref{eq:sc_max_fsm_state_general_new}). The corresponding expected number of erroneously added ones in the output stream can be calculated by
\begin{align}
  E_a = N P_{e,0}.
  \label{eq:exp_no_left_overflow}
\end{align}

For the error caused by the remaining ones in the shift register we compute the expected value of the shift register state,  
i.e. the expected number of ones in the shift register, as
\begin{align}
  E_r = \sum_{i=0}^{M-1} i P_i,
  \label{eq:exp_no_remaining_bits}
\end{align}

where $P_i$ denotes the probability of the FSM to be in the~$i$th~state. Note that the expected number of ones in the shift register corresponds to the number of ones that are missing on average in the output stream. 
Combining \eqref{eq:exp_no_left_overflow}  and \eqref{eq:exp_no_remaining_bits} and considering that left overflows add ones to the output stream, while the remaining bits in the shift register are the missing ones in the output stream, the expected number of erroneous ones is given by
\begin{align}
  E_{e, a > b} = E_a -E_r = N P_{e,0} - \sum_{i=0}^{M-1} i P_i.
\end{align}

Finally, the error probability can be expressed as
\begin{align}
  P_{e, a > b, N} = \frac{\left|E_{e, a < b}\right|}{N} = \left|P_{e,0} - \frac{1}{N} \sum_{i=0}^{M-1} i P_i\right|.
\label{eqn:Peba}
\end{align}

Interestingly, the second term in \eqref{eqn:Peba} depends on the stream length $N$, which goes to zero for $N \rightarrow \infty$. 
This is because a finite number of missing bits has a higher impact on the error for shorter streams than for longer streams.

\subsection{Expected Error Probability}
Assuming a uniform distribution for $a$ and $b$ over $[0,1]$ and using (\ref{eqn:Peab}) and (\ref{eqn:Peba}), the expected error probability can be derived as follows
\begin{align}
  E(P_{e, N})  = \int\limits_0^1  \left( \int\limits_0^b P_{e, a \leq b}  \text{d}a \right ) \text{d}b + \int\limits_0^1 \left( \int\limits_0^a P_{e, a > b,N}  \text{d}b \right ) \text{d}a.
  \label{eq:tot_prob_err}
\end{align}
The two integrals in \eqref{eq:tot_prob_err}, cover exactly half of the two dimensional space $[0,1] \times [0,1]$. Thus, 
they form the expected value over the whole space. 

Unfortunately, to the best of our knowledge, for this integral no closed-form solution exists. Thus, we calculated it through numerical integration.
Fig.~\ref{fig:ExpectedError} shows the error probabilities obtained through numerical integration and the simulation results. For each simulated point, the empirical error probability was averaged over $10000$ test cases. We observe a good match between the theoretical and the simulation results. However, the analytical results are obtained much faster than the simulation results. 
Moreover, it can be seen that for a certain stream length $N$ there exists an optimal shift register length $L_\text{opt}$. For example, for~$N=10^{4}$ the error probability decreases until length~$15$ and then increases due to the remaining bits 
in the shift register. Thus, for $N=10^{4}$ the optimal shift register length is given by~$L_\text{opt}=15$. In Fig.~\ref{fig:ExpectedError}, we marked the optimal shift register lengths with triangles and added extra ticks at the x-axis. The lower bound curve in Fig.~\ref{fig:ExpectedError} corresponds to the performance limit if the bit stream length $N$ goes to infinity. It can either be obtained by numerically integrating~(\ref{eqn:ExpMInf})~or~(\ref{eq:tot_prob_err}), the latter with $N \rightarrow \infty$. For the latter approach, the second term on the right hand side of (\ref{eqn:Peba}) goes to zeros, resulting in~$P_{e,a>b, N \rightarrow \infty} = P_{e,0}$. This is because when the stream length goes to infinity, a finite number of remaining bits in the shift register does not matter.

\begin{figure}[t]
\begin{center}
\vspace{-.03cm}
\begin{tikzpicture}
\begin{semilogyaxis}[compat=newest, 
width=.9\columnwidth, height =.7\columnwidth,log basis y=10, grid,
ylabel={ $E(P_e)$ }, 
xlabel={ Shift Register Length $L$}, 
ymin = 3e-5,
legend style={at={(0.845,1.0)}},
legend columns = {2},
ymax = .4,
legend cell align=left,
extra x ticks={6,15,22,27,34},
extra x tick labels={\scriptsize \hspace{-2mm} $6$, \scriptsize $15$,\scriptsize $22$,\scriptsize $27$,\scriptsize $34$},
extra x tick style={ grid=none,
xticklabel style={xshift = 5, yshift=12}},
clip marker paths=true
]
\draw[densely dotted, thin, color=gray] (6,4.1274801e-03) --   (6,4e-5) ;
\draw[densely dotted, thin, color=gray] (15,1.0299568e-03) -- (15,4e-5) ;
\draw[densely dotted, thin, color=gray] (22,5.1816724e-04) -- (22,4e-5) ;
\draw[densely dotted, thin, color=gray] (27,3.7486047e-04) -- (27,4e-5) ;
\draw[densely dotted, thin, color=gray] (34,2.4078136e-04) -- (34,4e-5) ;

\addplot[color=Red,very thick] table[x index =0, y index =1] {./p_only_1000_10000_30000_50000_100000_300000.dat};
\addlegendentry{ \scriptsize $N=1000$}
\addplot[color=Green,very thick] table[x index =0, y index =2] {./p_only_1000_10000_30000_50000_100000_300000.dat};
\addlegendentry{ \scriptsize $N=10000$}
\addplot[color=Blue,very thick] table[x index =0, y index =3] {./p_only_1000_10000_30000_50000_100000_300000.dat};
\addlegendentry{ \scriptsize $N=30000$}
\addplot[color=Black,very thick] table[x index =0, y index =4] {./p_only_1000_10000_30000_50000_100000_300000.dat};
\addlegendentry{ \scriptsize $N=50000$}
\addplot[color=MidnightBlue,very thick] table[x index =0, y index =5] {./p_only_1000_10000_30000_50000_100000_300000.dat};
\addlegendentry{ \scriptsize $N=100000$}

\addplot[color=Red,very thick, only marks, mark=+,forget plot] table[x index =0, y index =1] {./e_1000_10000_30000_50000_100000__10000.dat};
\addplot[color=Green,very thick, only marks, mark=+,forget plot] table[x index =0, y index =2] {./e_1000_10000_30000_50000_100000__10000.dat};
\addplot[color=Blue,very thick, only marks, mark=+,forget plot] table[x index =0, y index =3] {./e_1000_10000_30000_50000_100000__10000.dat};
\addplot[color=Black,very thick, only marks, mark=+,forget plot] table[x index =0, y index =4] {./e_1000_10000_30000_50000_100000__10000.dat};
\addplot[color=MidnightBlue,very thick, only marks, mark=+,forget plot] table[x index =0, y index =5] {./e_1000_10000_30000_50000_100000__10000.dat};
\node[regular polygon,regular polygon sides=3, fill,color=Red,inner sep=0pt,minimum size=2.5mm] at (axis cs:6,4.1274801e-03) {};
\node[regular polygon,regular polygon sides=3,fill,color=Green,inner sep=0pt,minimum size=2.5mm] at (axis cs:15,1.0299568e-03) {};
\node[regular polygon,regular polygon sides=3,fill,color=Blue,inner sep=0pt,minimum size=2.5mm] at (axis cs:22,5.1816724e-04) {};
\node[regular polygon,regular polygon sides=3,fill,color=Black,inner sep=0pt,minimum size=2.5mm] at (axis cs:27,3.7486047e-04) {};
\node[regular polygon,regular polygon sides=3,fill,color=MidnightBlue,inner sep=0pt,minimum size=2.5mm] at (axis cs:34,2.4078136e-04) {};

\addplot[color= 	Magenta,very thick, densely dotted] table[x index =0, y index =3] {./ApproxErrorAnalysisVarAscii_overSR_51.dat};
\addlegendentry{ \scriptsize lower bound }

\end{semilogyaxis}
\end{tikzpicture}
\caption{ Expected error probability versus shift register length $L$ for different bit stream lengths $N$. Solid lines: theoretical results (\ref{eq:tot_prob_err}); markers~(+): bit-wise simulation; markers ($\blacktriangle$): optimal shift register lengths~$L_\text{opt}$. \label{fig:ExpectedError}}
\end{center}
\end{figure}
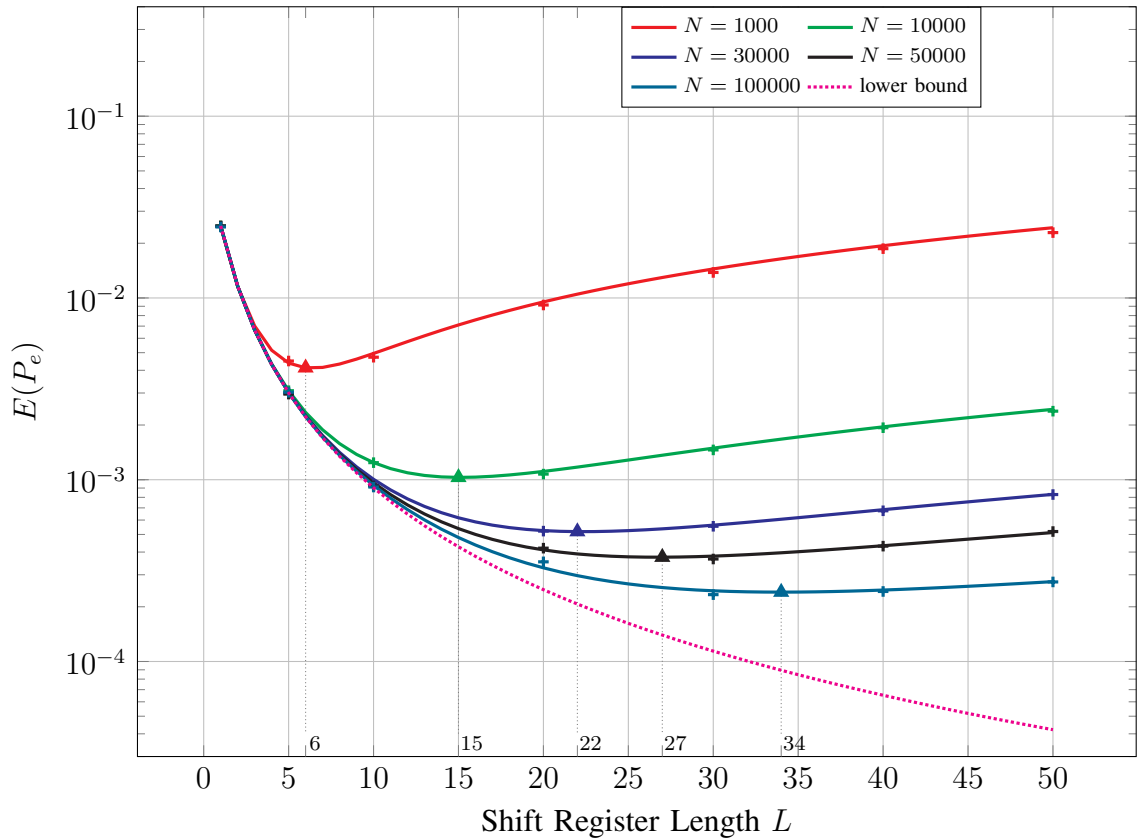

%% file: Conclusions.tex

In this work, we investigated the stochastic SMax/SMin function, which is 
an important building block in many applications (e.g., max pooling in neural networks). 
Prior works have proposed circuits for the SMax/SMin function and provided 
empirical validation. In this paper, we analytically proved the correctness of these architectures.
Moreover, we proposed a novel shift-register based SMax/SMin function, which outperforms the 
state-of-the-art architectures in terms of accuracy, while having comparable hardware cost.
We provided a new error analysis of the proposed circuit, considering the value of the individual 
bits in the stochastic stream. This analysis revealed that for practical bit stream lengths 
a finite optimal shift register length exists. Moreover, we showed that increasing the shift register length 
beyond the optimal value deteriorates the accuracy. This is due to the error caused by the remaining bits in the shift register. 
Hence, finding strategies to empty the shift register might be an interesting future extension of this work.

\balance